\documentclass[a4paper,UKenglish]{article}
\usepackage{amsthm}
\usepackage{microtype}
\usepackage{fullpage}
\usepackage{graphics}
\graphicspath{{./figures/},{./figures/}}
\usepackage{thm-restate}
    \newtheorem{theorem}{Theorem}
    
    \newtheorem{lemma}[theorem]{Lemma}

\usepackage{authblk}
\usepackage{graphicx,amsmath,xcolor,amssymb}
\usepackage{xspace}

\newcommand{\maxmincs}{Bottleneck Convex Subsets\xspace}
\newcommand{\ap}{Angle Partition\xspace}
 
\usepackage{hyperref}
\usepackage{verbatim}

\usepackage{framed}

\title{Bottleneck Convex Subsets: Finding $k$ Large\\ Convex Sets in a Point Set\thanks{The work is supported in part by the Natural Sciences and Engineering Research Council of Canada (NSERC). The authors dedicate the work in memory of Saeed Mehrabi.}}
\author[1]{Stephane Durocher}
\author[2]{J. Mark Keil}
\author[3]{Saeed Mehrabi}
\author[2]{Debajyoti Mondal} 
 
\affil[1]{Department of Computer Science, University of Manitoba, Winnipeg, Canada\\
  \texttt{durocher@cs.umanitoba.ca}}

\affil[2]{Department of Computer Science, University of Saskatchewan, Saskatoon, Canada\\
  \texttt{keil@cs.usask.ca}, \texttt{dmondal@cs.usask.ca}}  

\affil[3]{Carleton University, Ottawa, Canada}

\usepackage[textsize=tiny]{todonotes}
\usepackage{verbatim}

\begin{document}
\maketitle              
\begin{abstract}
Chv\'{a}tal and Klincsek (1980) gave an $O(n^3)$-time algorithm for the problem of finding a maximum-cardinality convex subset of an arbitrary given set $P$ of $n$ points in the plane. This paper examines a generalization of the problem, the {\em \maxmincs} problem: given a set $P$ of $n$ points in the plane and a positive integer $k$, select $k$ pairwise disjoint convex subsets of $P$ such that the cardinality of the smallest subset is maximized. Equivalently, a solution maximizes the cardinality of $k$ mutually disjoint convex subsets of $P$ of equal cardinality. We show the problem is NP-hard when $k$ is an arbitrary input parameter, we give an algorithm that solves the problem exactly, with running time polynomial in $n$ when $k$ is fixed, and we give a fixed-parameter tractable algorithm parameterized in terms of the number of points strictly interior to the convex hull. 
\end{abstract}

\section{Introduction}

A set $P$ of points in the plane is {\em convex} if for every $p \in P$ there exists a closed half-plane $H^+$ such that $H^+ \cap P = \{p\}$. Determining whether a given set $P$ of $n$ points in the plane is convex requires $\Theta(n\log n)$ time in the worst case, corresponding to the time required to determine whether the convex hull of $P$ has $n$ vertices on its boundary \cite{kirkpatrick1986}. Chv\'{a}tal and Klincsek \cite{chvatal1980} gave an $O(n^3)$-time and $O(n^2)$-space  algorithm to find a maximum-cardinality convex subset of any given set $P$ of $n$ points in the plane. Later, Edelsbrunner and Guibas~\cite{DBLP:journals/jcss/EdelsbrunnerG89} improved the space complexity to $O(n)$. 
In this paper, we examine a generalization of the problems to multiple convex subsets of $P$. Given a set $P$ of points in the plane and a positive integer $k$, we examine the problem of finding $k$ convex and mutually disjoint subsets of $P$, such that the cardinality of the smallest set is maximized.
We define the problem formally, as follows.

\begin{enumerate}   
\item[]{\bf BOTTLENECK CONVEX SUBSETS}

\item[]{\bf Instance:} A set $P$ of $n$ points in $\mathbb{R}^2$, and a positive integer $k$.

\item[]{\bf Problem:} Select $k$ sets $P_1, \ldots, P_k$ such that 
\begin{itemize}
    \item $\forall i \in \{1,\ldots k\},$ $P_i \subseteq P$,
    \item $\forall i \in \{1,\ldots k\},\ P_i$ is convex,
    \item $\forall \{i,j\} \subseteq \{1,\ldots k\}, \ i \neq j \Rightarrow P_i \cap P_j = \varnothing,$ and
    \item $\displaystyle\min_{i \in \{1, \ldots, k\}} |P_i|$ is maximized.
\end{itemize}
\end{enumerate}
\begin{figure}[pt]
\centering
\includegraphics[width=.9\textwidth]{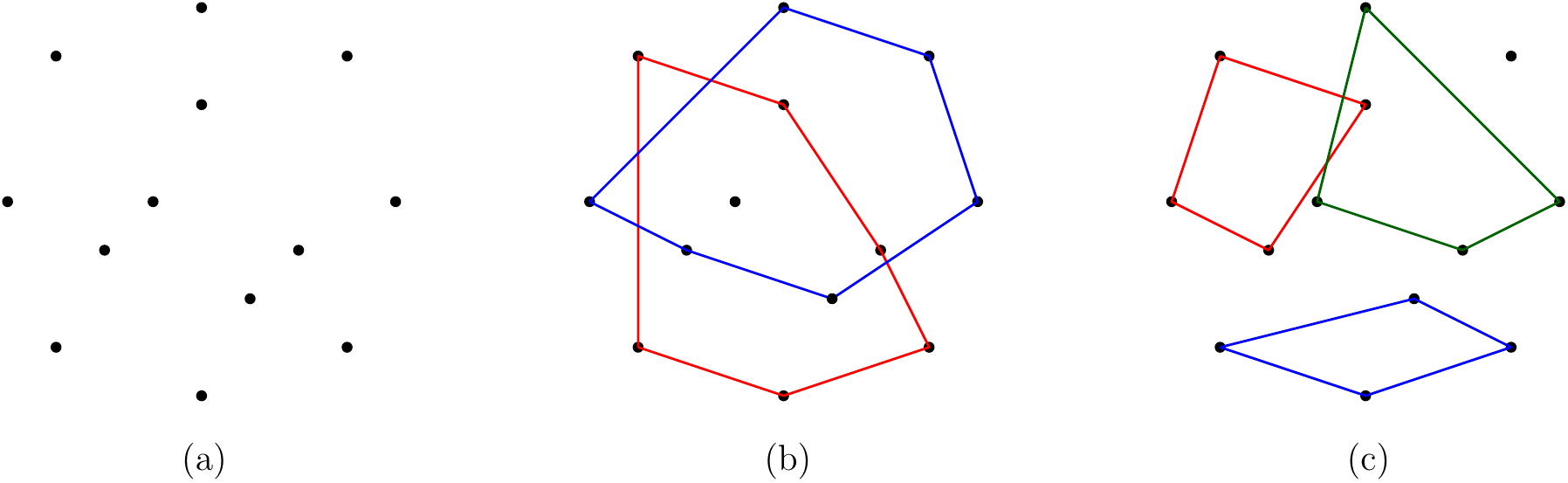}
\caption{(a) A point set $P$. (b) A solution to the \maxmincs problem when $k=2$. (c) A solution when $k = 3$. }
\end{figure}
Since every subset of a convex set of points remains convex, any $k$ convex sets can be made to have equal cardinality by removing points from any set whose cardinality exceeds that of the smallest set. Therefore, an equivalent problem is to find $k$ mutually disjoint convex subsets of $P$ of equal cardinality, where the cardinality is maximized. 

\subsection{Our Contributions}
In this paper we examine the problem of finding 
$k$ large convex subsets of a given point set with $n$ points. Our contributions are as follows:

\begin{enumerate}
    \item We give a polynomial-time algorithm that solves \maxmincs for any fixed $k$. The algorithm constructs a directed acyclic graph $G$ whose vertices correspond to distinct configurations of edges passing though vertical slabs between neighbouring points of $P$. A solution to the problem is found by identifying a node in $G$ associated with a maximum-cardinality set that is reachable from the source node.  
    \item Using a reduction from a restricted version of Numerical 3-Dimensional Matching, which is known to be NP-complete, we show that \maxmincs is NP-hard when $k$ is an arbitrary input parameter.
    \item We show that \maxmincs is fixed-parameter tractable when parameterized by the number of  points that are strictly interior to the convex hull of the given point set, i.e., the number of non-extreme points. Therefore, if the number of points interior to the convex hull is fixed, then for every $k$, \maxmincs can be solved in polynomial time.
\end{enumerate}

\subsection{Related Work}
\label{sec:related}

A \emph{convex $k$-gon} is a convex set with $k$ points. A convex \emph{$k$-hole} within a set $P$ is a convex $k$-gon on a subset of $P$ whose convex hull is empty of any other points of $P$. A rich body of research examines convex $k$-holes in point sets~\cite{Morris}.  
By the Erd\H{o}s-Szekeres theorem~\cite{erdos}, every point set with $n$ points in the Euclidean plan contains a convex $k$-gon for some $k \in \Omega(\log n)$.  Urabe~\cite{DBLP:journals/dam/Urabe96} showed that by repeatedly extracting such a convex $\Omega(\log n)$-gon, one can partition a point set into $O(n/\log n)$ convex subsets, each of size $O(\log n)$. 

Given a set $P$ of $n$ points in the plane, there exist $O(n^3)$-time algorithms to compute a largest convex subset of $P$~\cite{chvatal1980,DBLP:journals/jcss/EdelsbrunnerG89} and a largest empty convex subset of $P$~\cite{DBLP:conf/compgeom/AvisR85}. Both problems are NP-hard in $\mathbb{R}^3$~\cite{DBLP:conf/esa/GiannopoulosKW13}. In fact, finding a largest empty convex subset  is W[1]-hard in $\mathbb{R}^3$~\cite{DBLP:conf/esa/GiannopoulosKW13}.  Gonz\'{a}lez-Aguilar et al.~\cite{DBLP:conf/fct/Gonzalez-Aguilar19} have recently examined the problem of finding a largest convex set in the rectilinear setting.

The \emph{convex cover number} of a point set $P$ is the minimum number of disjoint convex sets that covers $P$. The \emph{convex partition number} of a point set $P$ is the minimum number of convex sets with disjoint convex hulls (in addition to their vertex sets being pairwise vertex disjoint) that covers $P$. Urabe~\cite{DBLP:journals/dam/Urabe96} examined lower and upper bounds on the convex cover number and the convex partition number. He showed that the convex cover number of a set of $n$ points in $\mathbb{R}^2$ is in $\Theta(n/\log n)$ and its convex partition number is bounded from above $\lceil \frac{2n}{7}\rceil$. Furthermore, there exist point sets with convex partition number at least $\lceil \frac{n-1}{4}\rceil$.

Arkin et al.~\cite{DBLP:journals/corr/cs-CG-0210003}  proved that both finding the convex cover number and the convex partition number of a point set are NP-hard problems, and gave a polynomial-time $O(\log n)$-approximation algorithm for both problems. Although the \maxmincs problem appears to be similar to the convex cover number problem as both problems attempt to find disjoint convex sets, the objective functions are different.  Neither the NP-hardness proof nor the approximation result for convex cover number~\cite{DBLP:journals/corr/cs-CG-0210003} readily extends to the \maxmincs problem. Previous work has also considered partitioning a point set  into empty convex sets, where the convex hulls of the sets do not contain any interior point.  For the number of empty convex point sets, an upper bound of 
$\lceil \frac{9n}{34}\rceil$ and a lower bound of $\lceil \frac{n+1}{4}\rceil$ is known~\cite{DBLP:conf/jcdcg/DingHUX02}. We refer the readers to~\cite{DBLP:journals/dcg/EppsteinE94,DBLP:journals/dcg/EppsteinORW92} for related problems on finding convex sets with various optimization criteria.

Another related problem in this context is to partition a given point set using a minimum number of lines (Point-Line-Cover), which Megiddo and Tamir~\cite{DBLP:journals/orl/MegiddoT82}  showed to be NP-hard, and was subsquently shown to  be APX-hard~\cite{DBLP:conf/cccg/BrodenHN01,DBLP:conf/icalp/KumarAR00}. 
Point-Line-Cover is known to be fixed-parameter tractable when parameterized on the number of lines. 
Whether the minimum convex cover problem is fixed-parameter tractable remains an open problem~\cite{Eppstein}. Note that for any fixed $k$, one can decide whether the minimum convex cover number of a point set is at most $k$ in polynomial time~\cite{DBLP:journals/corr/cs-CG-0210003}. 


Previous work on the Ramsey-remainder problem provides insight into the 
 \maxmincs problem~\cite{DBLP:journals/ejc/ErdosTV96}. Given an integer $i$, the Ramsey-remainder is the smallest integer $rr(i)$ such that for every sufficiently large point set, all but $rr(i)$ points can be partitioned into convex sets of size at least $i$. Therefore, a \maxmincs problem with sufficiently large $n$ and with $k\le \lfloor \frac{n - rr(k)}{k} \rfloor$ must have a solution where the size of the smallest convex set is at least $k$. Note that the \maxmincs problem is straightforward to solve for the case when $k \ge n/3$, i.e., one needs to compute a balanced partition without worrying about the convexity of the sets. However, the case when  $k = n/4$ already becomes nontrivial.  K\'{a}rolyi~\cite{DBLP:journals/dam/Karolyi01} derived a necessary and sufficient condition for a set of $4n$ points in general position to admit a partition into $n$ convex quadrilaterals, and gave an $O(n \log n)$-time algorithm to decide whether such a partition exists. 

\section{A Polynomial-Time Algorithm for a Fixed $k$}
Given a set $P$ of $n$ points in the plane and a fixed integer $k$, we describe an $O(k n^{5k+3})$-time algorithm that solves \maxmincs for any fixed $k$. The idea is to construct a directed acyclic graph $G$ whose vertices each correspond to a vertical slab of the plane in a given state with respect to the selected subsets $P_1, \ldots, P_k$ of $P$, with an edge from one slab to the slab immediately to its right if the states of the two neighbouring slabs form a locally mutually compatible solution. A feasible solution ($P_1, \ldots, P_k$ are mutually disjoint convex subsets of $P$) corresponds to a directed path starting at the root node in $G$, i.e., a sequence of consecutive compatible slabs. Among the feasible solutions, an optimal solution ($\min_{i \in \{1,\ldots, k\}} |P_i|$ is maximized) corresponds to a path that ends at a node for which the cardinality of the smallest set is maximized.  

Rotate $P$ such that no two of its points lie on a common vertical line. 
Partition the plane into $n-1$ vertical slabs, $S_1, \ldots, S_{n-1}$, determined by the $n$ vertical lines through points of $P$. Let $L$ be the set of ${n \choose 2}$ line segments whose endpoints are pairs of points in $P$.
Within each slab, $S_i$, consider the set of line segments $L_i = \{l \cap S_i \mid l \in L\}$. A convex point set corresponds to the vertices of a convex polygon; in a feasible solution, $j$ convex polygons intersect $S_i$ for some $j \in \{0, \ldots, k\}$. Each of these polygons has a top segment and a bottom segment in $L_i$. There are at most ${|L_i| \choose 2}$ possible choices of segments in $L_i$ for the first polygon, ${|L_i|-2 \choose 2}$ for the second polygon, $\ldots$, and ${|L_i|-2(j-1) \choose 2}$ for the $j$th polygon, giving
$\prod_{x=0}^{j-1} {|L_i|-2x \choose 2} \in O(|L_i|^{2j})= O(n^{4j})$ possible combinations of edges in $S_i$ for a given $j \in \{0, \ldots, k\}$.

We construct an unweighted directed acyclic graph $G$. Each vertex in $V(G)$ corresponds to a slab $S_i$, a $j \in \{0, \ldots, k\}$, and a top edge and a bottom edge for each of the $j$ convex polygons that intersect $S_i$. Consequently, the number of vertices in $G$ is $O(\sum_{i=1}^{n-1} \sum_{j=0}^k n^{4j}) = O(kn^{4k+1})$. 

Furthermore, we create $(n/k)^k$ copies of each vertex associated with a slab $S_i$, each of which is assigned a distinct value $(\ell_1, \ldots, \ell_k) \in \mathbb{Z}^k$, where for each $j\in \{1, \ldots, k\}$, $\ell_j = | P_j \cap (S_1 \cup \cdots \cup S_i)|$, i.e., the number of points of $P_j$ that lie in the first $i$ slabs. We refer to $\ell = \min_{j \in \{1, \ldots, k\}}\ell_j$ as the vertex's {\em level}. Each vertex at level $\ell$ in $G$ corresponds to a slab $S_i$, such that the minimum cardinality of any polygon in $S_1 \cup \ldots \cup S_i$ (or partial polygon if it includes points to the right of $S_i$) is $\ell$. Therefore, the resulting graph $G$ has $O((n/k)^k kn^{4k+1}) \subseteq O(\frac{1}{k^{k-1}}\cdot n^{5k+1})$ vertices. See Figure~\ref{fig:graphStructure}. 

\begin{figure}[pt]
\centering
\includegraphics[width=0.5\textwidth]{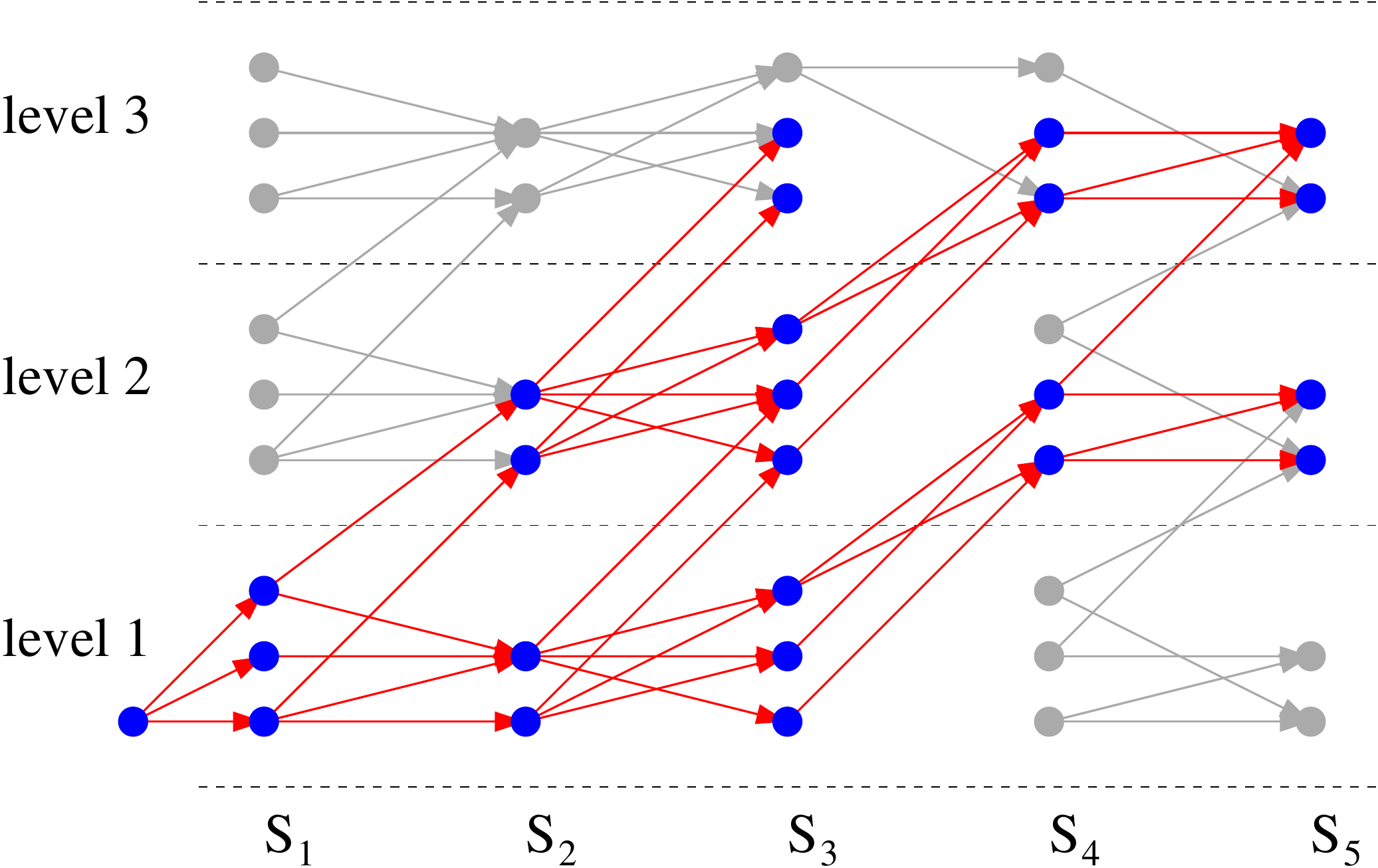}
\caption{Each slab $S_i$ has various combinations of pairs of edges possible, each of which corresponds to a vertex in $G$, which is copied at levels 1 through $n/k$. Directed edges are added from a vertex associated with slab $S_i$ to a vertex associated with a compatible slab $S_{i+1}$. The edge remains at the same level if the cardinality of the smallest set in $S_1 \cup \cdots \cup S_{i+1}$ remains unchanged; the level of $S_{i+1}$ is one greater than the level of $S_i$ if the cardinality of the smallest set in $S_1 \cup \cdots \cup S_{i+1}$ increases. Some vertices cannot be reached by any path from any source node at level 1 in slab $S_1$; these vertices and their out-edges are shaded gray. A feasible solution corresponds to a path rooted at a source node associated with the slab $S_1$ on level 1. An optimal solution ends at a sink node at the highest level among all feasible solutions.}
\label{fig:graphStructure}
\end{figure}

Every slab has exactly one point of $P$ on its left boundary and one on its right boundary. For each vertex $v$ in $G$, let $v_l$ and $v_r$ denote these two points of $P$ for the slab corresponding to $v$. We add an edge from vertex $u$ to vertex $v$ in $G$ if they are {\em compatible}. See Figure~\ref{fig:graphCases}. The vertices  $u$ and $v$ are compatible if:
\begin{itemize}
    \item $u$ and $v$ correspond to neighbouring slabs, $u$ to $S_i$ and $v$ to $S_{i+1}$, for some $i$, and
    \item all top and bottom segments associated with  $u$ that do not pass through $p_i$ continue in $v$, where $p_i = u_r = v_l$ is the point of $P$ on the common boundary of $S_i$ and $S_{i+1}$, and
    \item one of the four following conditions is met:
    \begin{enumerate}
        \item[Case 1.] 
    either (a) one top associated with $u$ ends at $p_i$ and one top associated with $v$ begins at $p_i$, forming a right turn at $p_i$, or (b) one bottom associated with $u$ ends at $p_i$ and  one bottom associated with $v$  begins at $p_i$, forming a left turn at $p_i$ (all polygons in $S_i$ continue in $S_{i+1}$; the number of edges in $S_i$ is equal to that in $S_{i+1}$); 
        \item[Case 2.]
        one top and one bottom associated with $u$ end at $p_i$, (one polygon ends in $S_i$ and all remaining polygons continue into $S_{i+1}$);
        \item[Case 3.]
    no top or bottom associated with  $u$ end at $p_i$, but one top and one bottom associated with $v$ start at $p_i$ (one polygon starts in $S_{i+1}$ and all remaining polygons continue from $S_i$ into $S_{i+1}$).
       \item[Case 4.]
    all edges in $u$ continue into $v$ and no edge passes through $p_i = u_r = v_l$ (all polygons in $S_i$ continue into $S_{i+1}$; the number of edges in $S_i$ is equal to that in $S_{i+1}$).
\end{enumerate}
\end{itemize}

\begin{figure}[pt]
\centering
\includegraphics[width=0.24\textwidth]{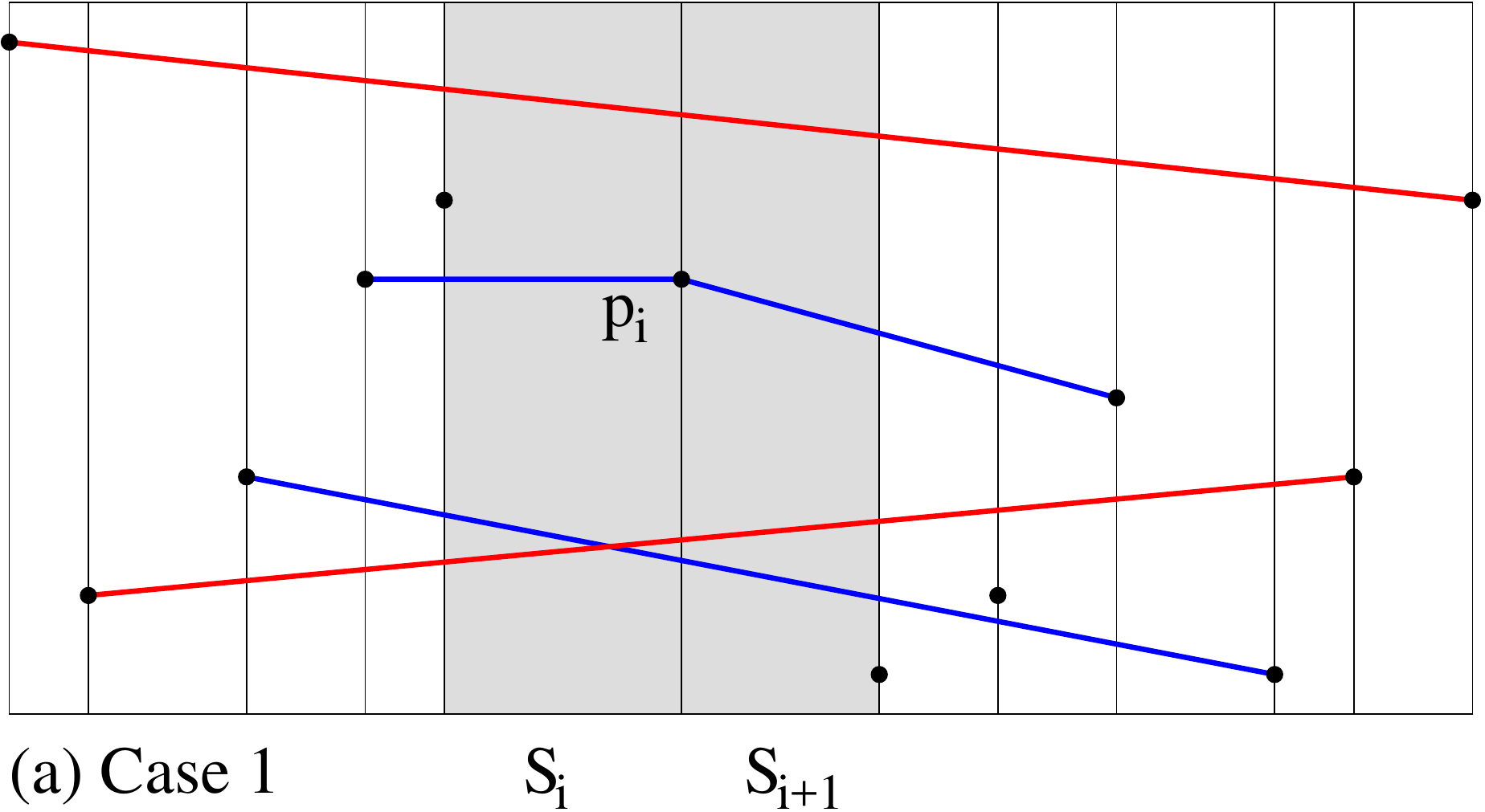}
\hfill
\includegraphics[width=0.24\textwidth]{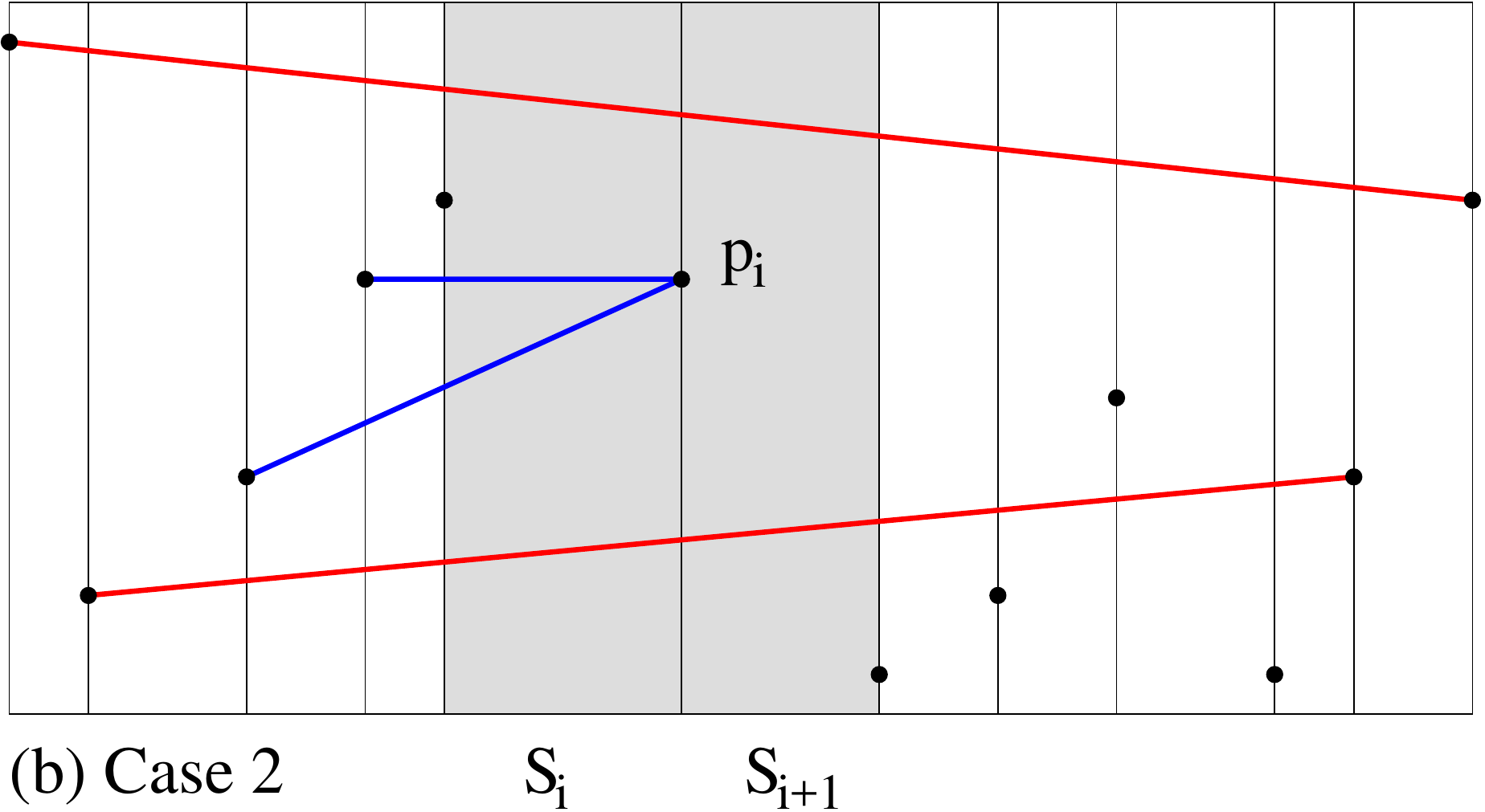}
\hfill
\includegraphics[width=0.24\textwidth]{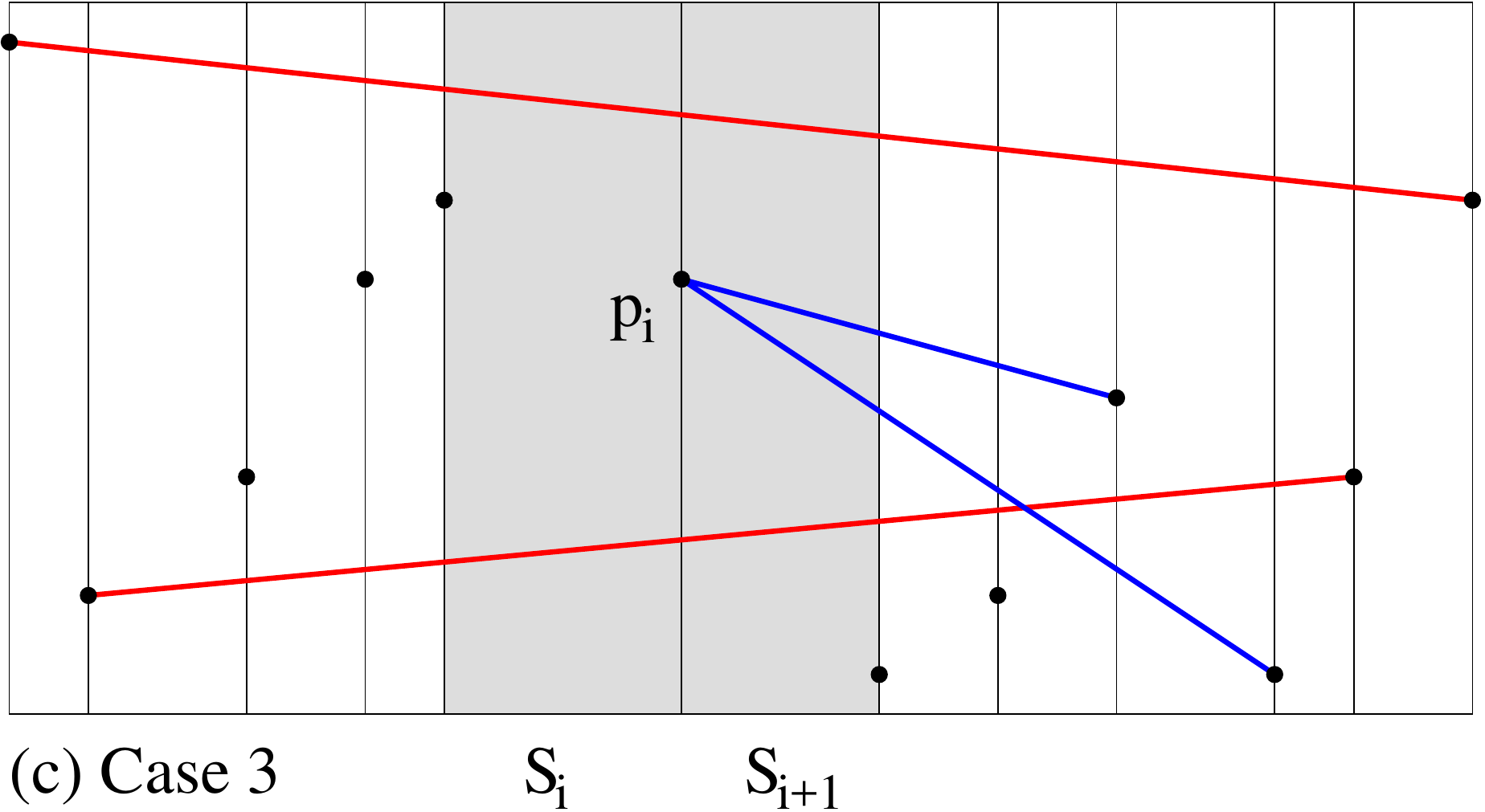}
\hfill
\includegraphics[width=0.24\textwidth]{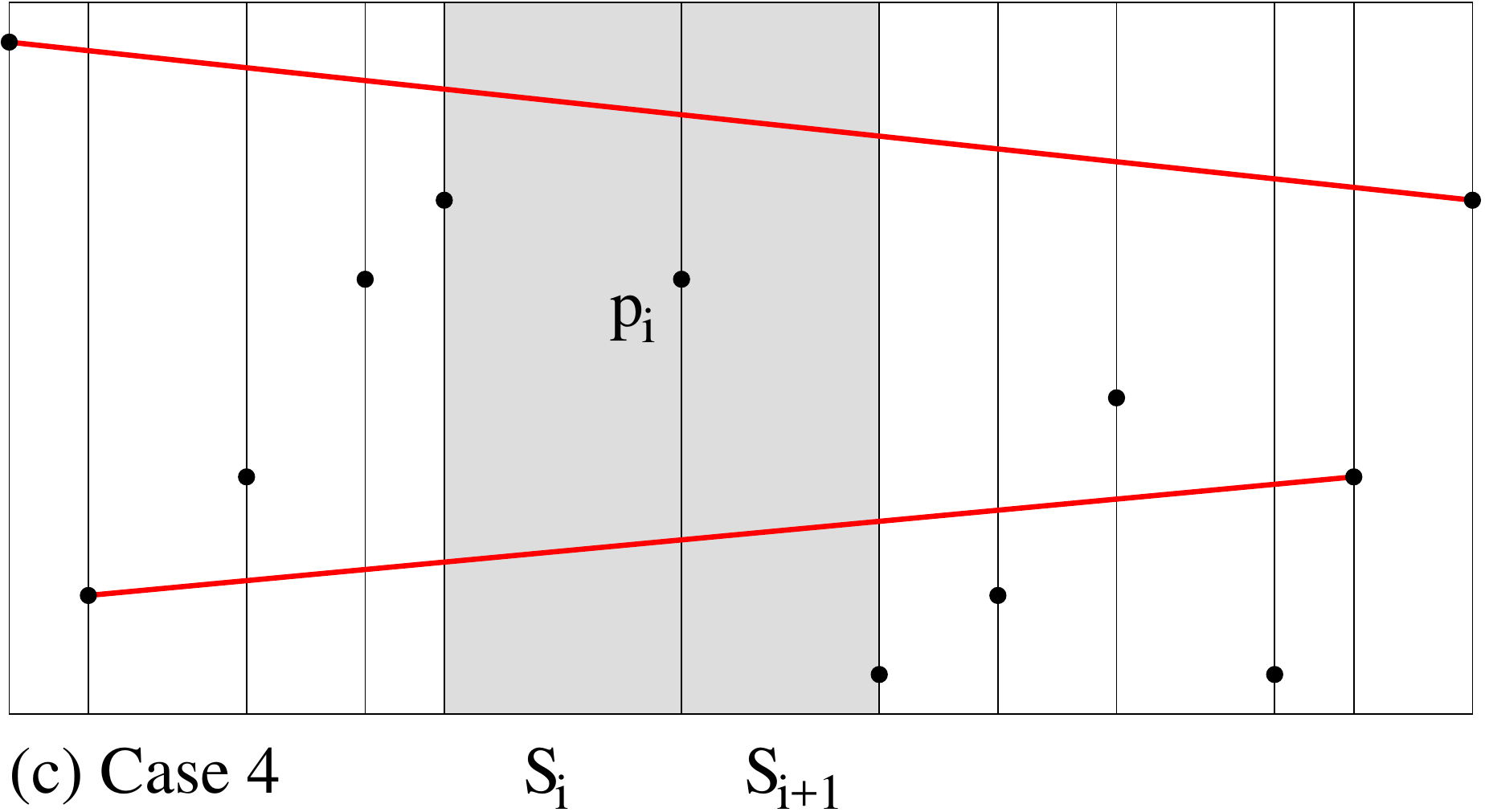}
\caption{The four cases in which we add an edge between the vertices $u$ (associated with the slab $S_i$) and $v$ (associated with the slab $S_{i+1}$) in $G$; i.e., $u$ and $v$ are {\em compatible}. In this example, $k=2$, corresponding to two polygons, for which the edges through $S_i$ and $S_{i+1}$ are coloured blue and red, respectively. In Figure~\ref{fig:graphCases}(a), $p_i$ lies on the upper hull of the blue polygon, so the polygon makes a right turn at $p_i$, i.e., the angle below $p_i$ must be convex. Figure~\ref{fig:graphCases}(d), $p_i$ is omitted from the selection.}
\label{fig:graphCases}
\end{figure}

For a given vertex $u$ at most $n-2$ edges satisfy Case~1 (there are at most $n-2$ possible edges that continue from $p_i$ to form a convex bend), at most one edge satisfies Case~2, at most ${n-3 \choose 2}$ edges satisfy Case~3, and at most one edge satisfies Case~4. Consequently, the number of edges in $G$ is $O(n^2 |V(G)|) \subseteq O(\frac{1}{k^{k-1}}\cdot n^{5k+3})$. 

Any path from a source on level 1 to a highest-level node corresponds to an optimal solution, and can be found using breadth-first search in time proportional to the number of edges in $G$. The resulting worst-case running time is proportional to the number of vertices and edges in $G$: $O(|V(G)| + |E(G)|) = O(\frac{1}{k^{k-1}}\cdot n^{5k+3})$. 
In addition to storing a single in-neighbour from which a longest path reaches each node $u$, we can maintain a list of all of its in-neighbours that give a longest path, allowing the algorithm to reconstruct all distinct optimal solutions with the running time increased only by the output size. 

The time for constructing the graph $G$ is proportional to its number of edges. The combinations of ${n \choose 2j}$ line segments in a slab $S_i$ on level $j$ can be enumerated and created in $O(1)$ time each, with $O(1)$ time per edge added if graph vertices are indexed according to their slab, their level, and the line segments they include. The level of each node in $G$ is determined in $O(1)$ time per node by examining the level of any of its in-neighbours; the level increases by one in Cases~1 and~2 if the point $p_i$ is added to the minimum-cardinality set and that set is the unique minimum. 

\begin{theorem}
Given a set $P$ of $n$ points in the plane, and a positive integer $k$, \maxmincs can be solved exactly in $O(\frac{1}{k^{k-1}}\cdot n^{5k+3})$ 
time. 
\end{theorem}

\section{NP-Hardness}
In this section we show that \maxmincs is NP-hard. 
We first introduce some notation.   
Let $x(p),y(p)$ be the $x$ and $y$-coordinates of a point $p$. An angle $\angle pqr$ determined by points  $p, q$ and $r$ is called a \emph{$y$-monotone angle}  if $y(p)>y(q)>y(r)$, as illustrated in Figure~\ref{fig:angle}. A $y$-monotone angle is \emph{left-facing} (resp. \emph{right-facing}) if the point $q$ lies interior to the left (resp., right) half-plane of the line through $pr$. If $q$ lies on the line through $pr$, then we refer to $\angle pqr$ as a \emph{straight angle}.  

The idea of the hardness proof is as follows. We first prove that given a set of $3n$ points in the Euclidean plane, it is NP-hard to determine whether the points can be partitioned into 
$n$  $y$-monotone angles, where none of them are right facing (Section~\ref{sec:angle}). We then reduce this problem to \maxmincs (Section~\ref{sec:hard}).

\begin{figure}[h]
\centering
\includegraphics[width=.6\linewidth]{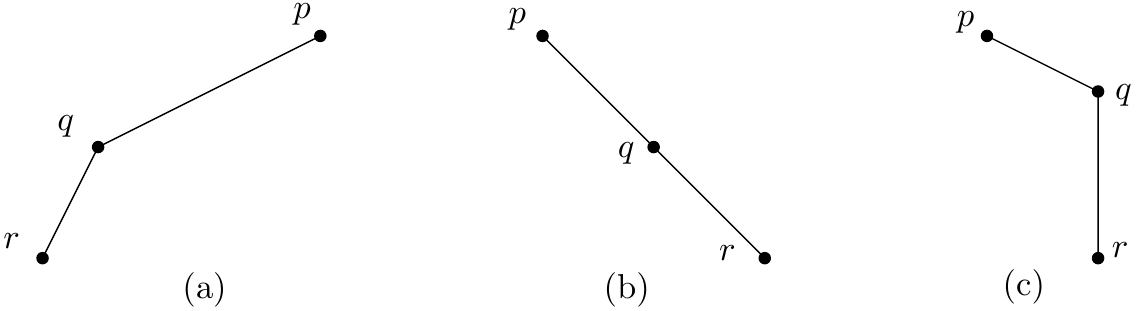}
\caption{Illustration for different types of $y$-monotone angles: (a) a left-facing angle, (b) a straight angle, and (c) a right-facing angle.}
\label{fig:angle}
\end{figure}

\subsection{Covering Points by Straight or Left-Facing Angles}
\label{sec:angle}
In this section we show that given a set of $3n$ points in the Euclidean plane, it is NP-hard to determine whether the points can be partitioned into $n$ $y$-monotone angles, where none of them are right facing. In fact, we prove the problem to be NP-hard in a restricted setting, 
as follows:

\begin{enumerate}   
\item[]{\bf ANGLE PARTITION}

\item[]{\bf Instance:} A set $P$ of $3n$ points lying on three parallel horizontal lines ($y=0, y=1$ and $y=2$) in the plane, where  each line contains exactly $n$ points.

\item[]{\bf Problem:} Partition $P$ into at most $n$ $y$-monotone angles, where none of them are right facing.
\end{enumerate}

We reduce Distinct 3-Numerical Matching with Target Sums (DNMTS), which is known to be strongly NP-complete~\cite[Corollary 8]{DBLP:journals/orl/HulettWW08}. 

\begin{enumerate}   
\item[]{\bf DISTINCT NUMERICAL MATCHING WITH TARGET SUM}

\item[]{\bf Instance:} Three sets $A=\{a_1,\ldots,a_n\},B=\{b_1,\ldots,b_n\},C=\{c_1,\ldots,c_n\}$, each with $n$ distinct positive integers, where $\sum_{i=1}^{n}a_i + \sum_{i=1}^{n}b_i = \sum_{i=1}^{n}c_i$.

\item[]{\bf Problem:} Decide whether there exist $n$ triples $(a_i,b_j,c_k)$, where $1\le i,j,k\le n$, such that $a_i+b_j =c_k$ and no two triples share an element.
\end{enumerate}


\begin{theorem}
\label{thm:ap}
\ap is NP-hard.
\end{theorem}
\begin{proof}
Let $M = (X,Y,Z)$ be an instance of DNMTS, where each set $A,B,C$ contains $n$ positive integers. We now construct an instance $Q$ of \ap as follows: (I)  For each $a\in A$, create a point at $(a,0)$. (II) For each $b\in B$, create a point at $(b,2)$. (III) For each $c\in C$, create a point at $(c/2,1)$.


This completes the construction of the point set $P$ of the \ap  instance $Q$ (e.g., see Figure~\ref{fig:ap}(a)). Since the numbers in $A,B,C$ are distinct, no two points in $P$ will coincide. Note that by definition, a $y$-monotone angle must contain one point from each of the lines $y=0, y=1$ and $y=2$. Furthermore, every straight angle $\angle pqr$ will satisfy the equation $\frac{x(p)+x(r)}{2}  = x(q)$. This transformation is inspired by a 3-SUM hardness proof for a  geometric problem known as `GeomBase'~\cite{DBLP:journals/comgeo/GajentaanO12}. 

We now show that $M$ has an affirmative solution if and only if $P$ admits a partition into $n$ $y$-monotone angles where none of them are right facing. 

First consider that $M$ has an affirmative answer, i.e., a set of $n$ triples $(a_i,b_j,c_k)$, where $1\le i,j,k\le n$, such that $a_i+b_j =c_k$ and no two triples share an element. Therefore, we will have $\frac{(a_i+b_j)}{2}  = \frac{c_k}{2}$. Hence we will find a straight line through $(a_i,0), (b_k,2), (c_j/2,1)$. These lines will form $n$ $y$-monotone straight angles (e.g., see Figure~\ref{fig:ap}(b)). Since none of these angles  are right facing, this provides an affirmative solution for the instance $Q$.

Consider now the case when $Q$ has an affirmative solution $T$, i.e., a partition of $P$ into $n$ $y$-monotone angles, where none of them are right facing.  We first claim that (Step 1) all these $n$  $y$-monotone angles must be straight angles and then (Step 2) show how to construct an affirmative solution for $M$. 

\begin{figure}[pt]
\centering
\includegraphics[width=\linewidth]{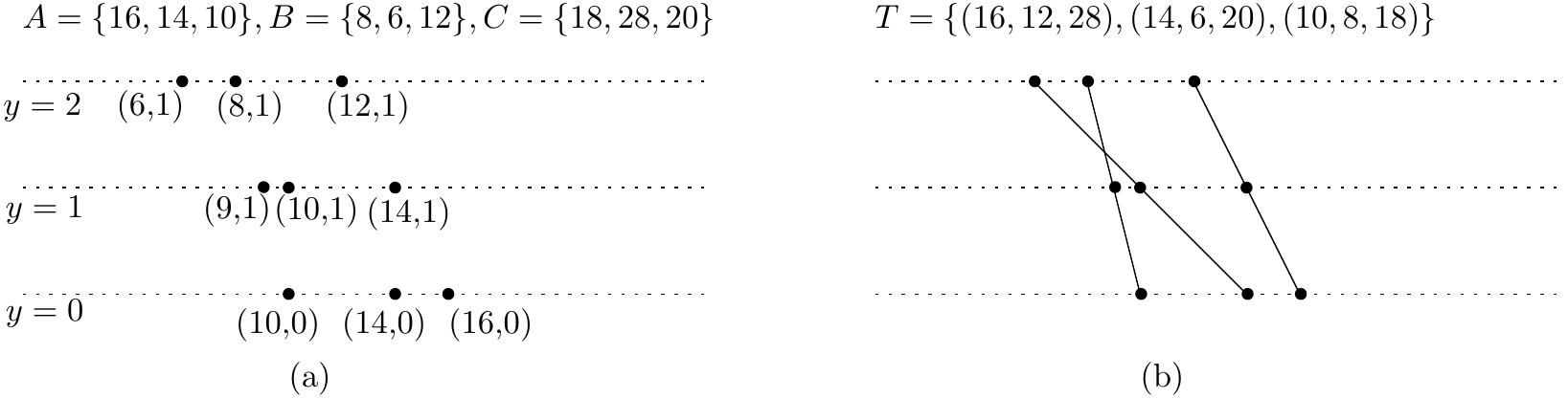}
\caption{(a) Construction of $Q$ from an instance $M$ of DNMTS. (b) A solution for $M$ and the corresponding angles of $Q$.}
\label{fig:ap}
\end{figure}

{\it Step 1:} Suppose for a contradiction that the solution $T$ contains one or more  left-facing angles.  For each left-facing angle  $\angle rst$, where $r,s,t$ are on lines $y=0,y=1$ and $y=2$, respectively,  we have $x(s) < \frac{x(r)+x(t)}{2}$. For each straight angle $\angle rst$, we have  $x(s) = \frac{x(r)+x(t)}{2}$. Since we do not have any right-facing angle, the following inequality holds: 
$\sum_{\angle rst \in T} x(s) < \sum_{\angle rst \in T}\frac{x(r)}{2}  + \sum_{\angle rst \in T} \frac{x(t)}{2}$.  
Since no two angles share a point, we have $\sum_{i=1}^{n} (c_i/2) < \sum_{i=1}^{n} (a_i/2) + \sum_{i=1}^{n} (b_i/2)  $, which contradicts that $M$ is an affirmative instance of DNMTS.

{\it Step 2:} We now transform the $y$-monotone straight angles of $T$ into $n$ triples for $M$. For each angle, $\angle rst$, where $r,s,t$ are on lines $y=0,y=1$ and $y=2$, we construct a triple $(x(r), x(t), 2x(s))$. Since $\angle rst$ is a straight angle,  $x(r)+ x(t) = 2x(s)$. Since no two angles share a point, the triples will be disjoint.  
\end{proof}

\subsection{\maxmincs is NP-Hard}
\label{sec:hard}
In this section we reduce \ap to \maxmincs.  Let $P$ be an instance of \ap, i.e., three lines $y=0,y=1$ and $y=2$, each line containing $n$ points. We construct an instance $H$ of \maxmincs with $k=n$.

\textbf{Construction of $H$:} We first take a copy $P'$ of the points of $P$ and include those in $H$. Let $\Delta$ be a sufficiently large number (to be determined later). We now construct $n$ \emph{upper chains}. The $i$th upper chain $U_i$, where  $1\le i\le n$, is constructed following the step below (see Figure~\ref{fig:kp}).
\begin{enumerate}
    \item[] \textit{Construction of $U_i$:} Place two points at the coordinates  $(i\Delta,\Delta^2+3)$ and $((i+1)\Delta,3)$. Let $C$ be the curve determined by $y=\Delta^2+3-(x-i\Delta)^2$, which passes through these two points. 
    Place $2n$ points uniformly on $C$ between $(i\Delta,\Delta^2+3)$ and $((i+1)\Delta,3)$.
\end{enumerate}
Each upper chain contains $(2n+2)$ points. We define the $n$ lower chains symmetrically, where each lower chain $V_i$ starts at $(i\Delta,-\Delta^2-1)$ and ends at $((i+1)\Delta,-1)$. 

We now choose the parameter $\Delta$. Let $t$ be the maximum $x$-coordinate of the points in $P$, and set  $\Delta$ to be $t^4$. This ensures that for any line $\ell$ with non-zero slope  passing through two points of $P$,  the upper and lower chains lie on the right half-plane of $\ell$. 

\begin{figure}[pt]
\centering
\includegraphics[width=.9\linewidth]{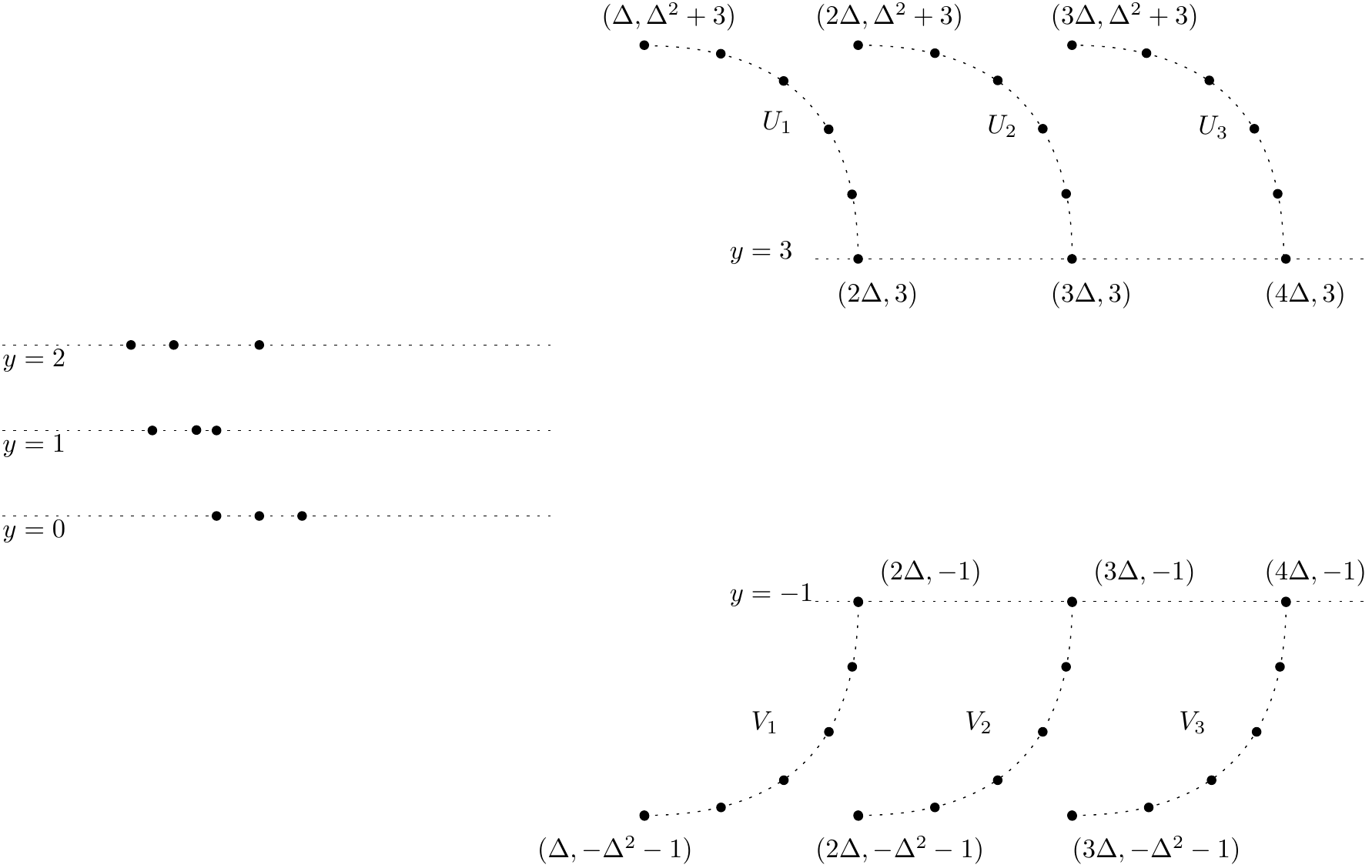} 
\caption{Illustration for the construction of $H$. Note that this is only a schematic representation, which violates the property that all the chains are inside the wedge determined by the $y$-monotone angles.}
\label{fig:kp}
\end{figure}

This concludes the construction of the \maxmincs instance $H$, where $k=n$. Note that $H$ has $3n+n(4n+4) = n(4n + 7)$ points. In the best possible scenario, one may expect to cover all the points and have a partition into $n$ disjoint convex subsets, where each set contains  $(4n+7)$ points. 

\begin{lemma}
\label{lem:both}
 Let $W$ be a partition of the upper and lower chains into  a set $L$ of at most $n$ disjoint convex sets. Then each convex set in $L$  contains at least one point from an upper chain and one point from a lower chain.
\end{lemma} 
\begin{proof} 
Suppose for a contradiction that we have a convex set that contains points from the same  type of chains, without loss of generality, from lower chains. Then we could delete all the points on the lower chain to obtain a convex set partition for the upper chains with fewer than $n$ disjoint convex sets. To reach the contradiction, we now  show that the upper chains cannot be covered with fewer than $n$ disjoint convex sets. 
 
Every three points of an upper chain forms a right-facing $y$-monotone angle. Since a convex set cannot have two such right-facing angles, no convex set can take three or more points from two different upper chains. 
 Since an upper chain $U$ contains $(2n+2)$ points, at least one convex set $C$ must contain at least 3  or more points from this set. We assign $C$ to $U$ and repeat this process for the other upper chains. Since $C$ cannot contain 3 points from any other upper chain, $C$ will not be assigned to any  upper chain except for $U$. Since each upper chain is assigned a unique convex set, we must have at least $n$ convex sets.
\end{proof}

\textbf{Reduction:} We now show that the \ap instance $P$ admits an affirmative solution if and only if the \maxmincs instance $H$ admits   $k(=n)$ disjoint convex sets with each set containing $(4n+7)$ points.   

Assume first that $P$ admits an affirmative solution, i.e., $P$ admits a set of $n$ $y$-monotone angles such that none of these are right facing. By the construction of $H$, the corresponding point set $P'$ must have such a partition into $y$-monotone angles. For each $i$ from 1 to $n$, we now form a point set $C_i$ that contains the $i$th $y$-monotone angle, the upper chain $U_i$ and the lower chain $V_i$.    Figure~\ref{fig:kp3}   illustrates such a scenario. 
By the construction of $H$, all the chains are inside the wedge determined by the $y$-monotone angle  and hence $C_i$ is a convex set with $(4n+7)$ points. Since the sets are disjoint, we obtain the required solution to the \maxmincs instance.

\begin{figure}[h]
\centering
\includegraphics[width=.9\linewidth]{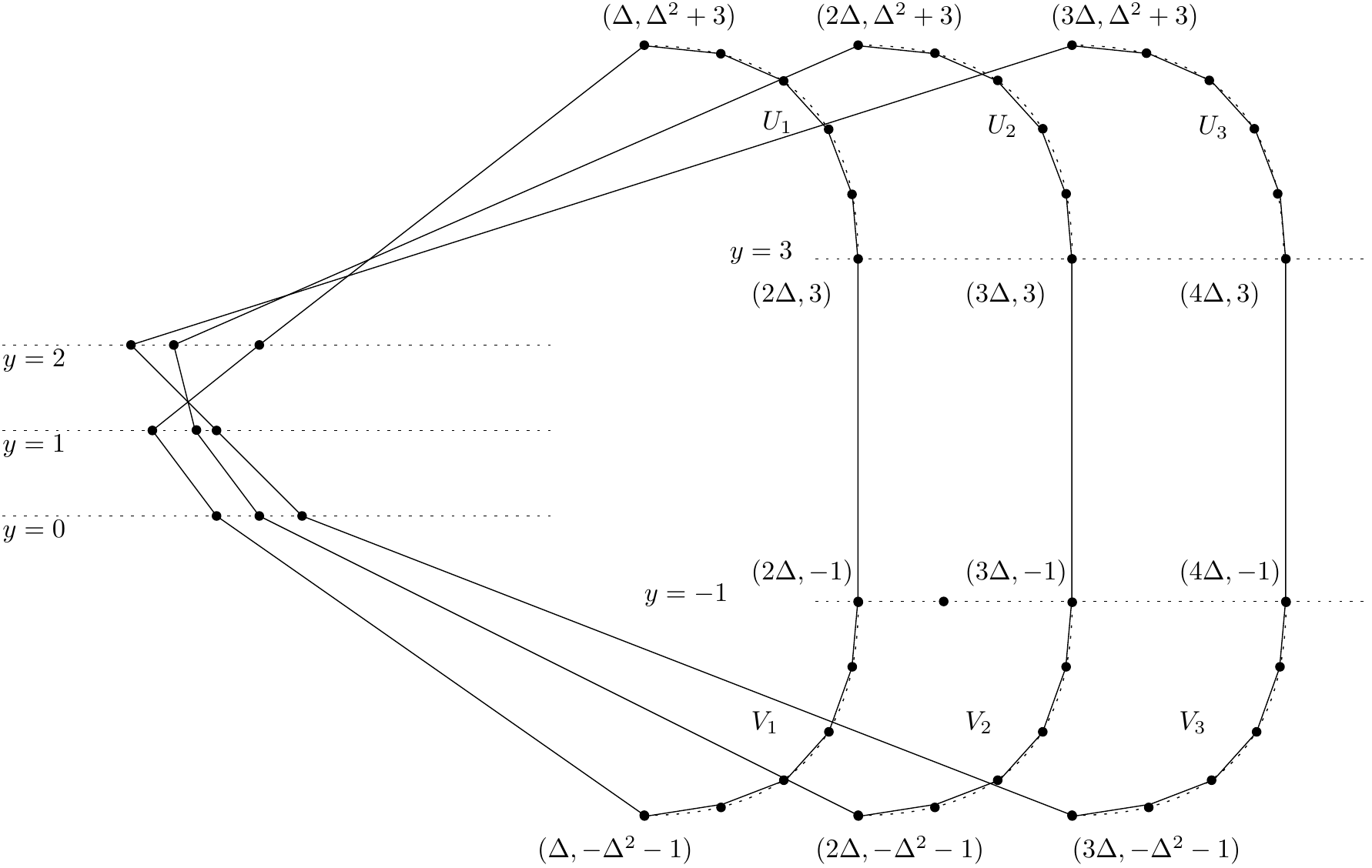}
\caption{A schematic representation for the construction of a convex partition for $H$ from an angle partition of $P$.}
\label{fig:kp3}
\end{figure}
Consider now that the points of $H$ admits   $n$ disjoint convex sets with each set containing $(4n+7)$ points. Since $H$ contains $n(4n+7)$ points, the convex sets form a partition of $H$.  Let $L$ be such a partition. We now show how to construct a solution for $P$ using $L$. Let $L'$ be a set of convex sets obtained by removing the points of $P'$ from each convex set of $L$. By Lemma~\ref{lem:both}, each set of $L'$ contains at least one point from the upper chains and one point from the lower chains.  
Since there are $3n$ points on $P'$, to partition $P'$ into $n$ convex sets, we must need each convex set of $L$ to contain 
a $y$-monotone angle with exactly one point from $y=0$, one point from $y=1$ and one point from $y=2$. Since each convex set contains one point from an upper chain  and one point from a lower chain, none of these $y$-monotone angles can be right facing. Hence we obtain a partition of $P'$ into the required $y$-monotone angles, which implies a partition also for $P$. This completes the reduction. 
The following theorem summarizes the results of this section.  
\begin{theorem}
The \maxmincs problem is NP-hard.
\end{theorem}

\section{Point Sets with Few Points inside the Convex Hull}
In this section we show that the \maxmincs problem is fixed-parameter tractable when parameterized by the number of points $r$ inside the convex hull, i.e., these points do not lie on the convex-hull boundary. 

\begin{theorem}
\label{thm:fpt}
Let $P$ be a set of $n$ points and let $r$ be the number of points interior to the convex hull of $P$. Then one can solve  the \maxmincs problem on $P$ in $f(r)\cdot n^{O(1)}$ time, i.e., the \maxmincs problem is fixed-parameter tractable when parameterized by $r$.  
\end{theorem} 
\begin{proof} 
Let $k$ be the number of disjoint convex sets that we need to construct. We guess the cardinality of the smallest convex set in an optimal solution and perform a binary search. 

For a guess $q$, we check whether there exists $k$ disjoint convex sets each with  $q$ points as follows. 

 \begin{figure}[h]
\centering
\includegraphics[width=.9\textwidth]{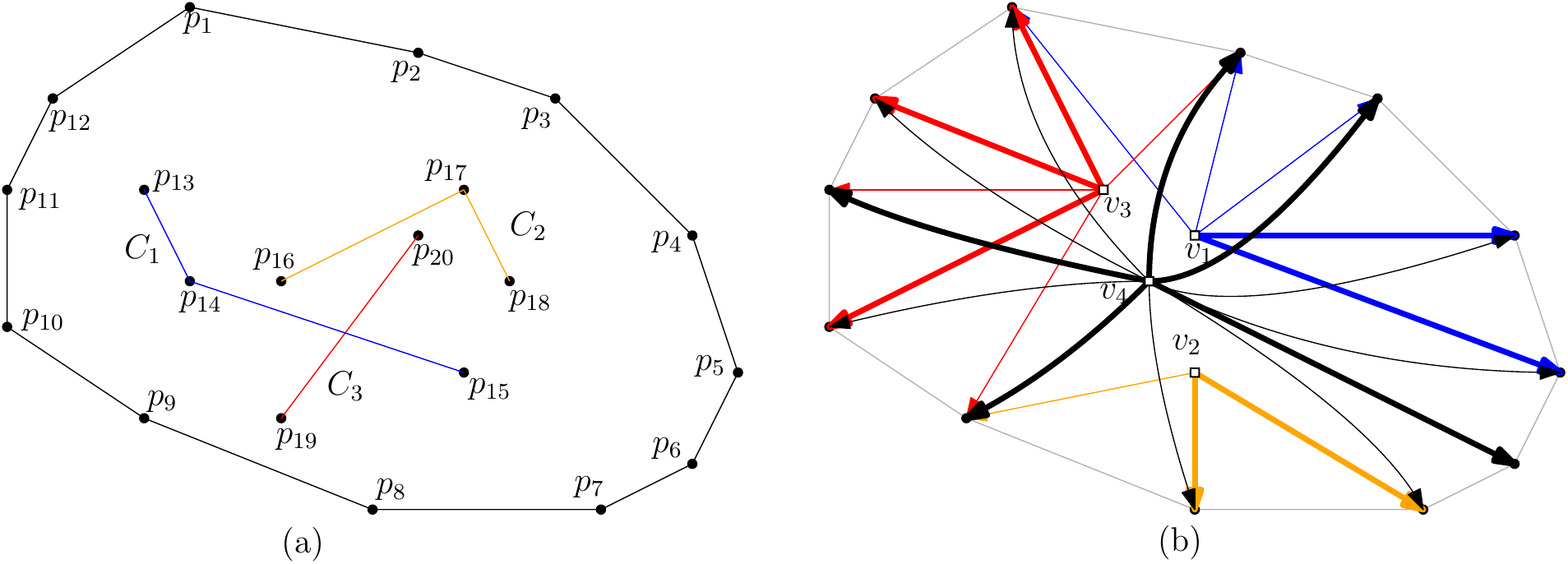}
\caption{Illustration for the \maxmincs problem with  eight points inside the convex hull of $P$. For the convex set corresponding to $v_3$, we have assigned the left halfplane of the line through $p_{19}$ and $p_{20}$. The edges carrying the flow are shown in thick edges.}
\label{fig:fpt}. 
\end{figure} 

Assume that $j$ of the $k$ convex sets contain points from the interior. Since there are only $r$ interior points, we  enumerate for each $j$ from 0 to $r$, all possible $j$ convex sets such that each set in these $j$ convex sets contains at most $q$ points from the interior of $P$. For each set of length $\ell  \le r$, we also consider all $\ell$ possible convex orderings of the points.  Figure~\ref{fig:fpt}(a) illustrates such a set of $j=3$ convex sets $C_1,C_2,C_3$ with a particular ordering of the points. Therefore, we have $\sum_{j=0}^k r {{2^r}\choose j}$ possibilities to consider. We need an additional consideration when all the points of a convex set lie on a straight line $L$. In that situation, we enumerate two further cases one that considers the left halfplane and the other that considers the right halfplane of $L$.  Thus the number of elements in the enumeration is at most   $\sum_{j=0}^k r {{2^r}\choose j}2^j \le \sum_{j=0}^k r 2^{r^{j+1}} \le  r 2^{r^{k+2}} $.

The idea is to examine whether these $j$ sets can be extended to contain $q$ points each and to check whether  the remaining points can be used to construct the remaining $(k-j)$ convex sets. To verify this, we construct a bipartite graph $G$ with vertex set $A\cup B$. The set $A$ contains $j$ vertices $v_1,\ldots,v_j$  corresponding to the sets $C_1,\ldots,C_j$ and $(k-j)$ additional vertices representing the remaining $(k-j)$ sets (which are currently empty) to be constructed. The set $B$ consist of $(n-r)$  vertices each corresponding to a distinct point on the convex hull of $P$. We add a directed edge from a vertex $v$ in $A$ to a vertex $w$ in $B$ if the point $w$ together with the  interior points corresponding to $v$ form a convex set. For the case when the interior points corresponding to $v$ form a straight line (e.g., $C_3$ in Figure~\ref{fig:fpt}(a)), we connect $v$ to the points of $B$ that lie on the halfplane assigned to $v$. Figure~\ref{fig:fpt}(b) illustrates the resulting graph.

We now consider a maximum flow on this graph where each vertex $v_i$ in $A$ has a production of $(q-|C_i|)$ units of flow and each sink can consume at most 1 unit of flow.  A maximum flow of $\sum_{j=0}^k (q-|C_i|)$ units indicates that the guess $q$ is feasible, and we continue the binary search by guessing a higher value. Otherwise, we search by guessing a lower value.

Hence the overall time complexity becomes $O(f(r)\cdot g(n)\log n)$, where $f(r) \in O(r 2^{r^{k+2}})$, $g(n)$ is the time required for the maximum flow algorithm, and the $\log n$ term corresponds to the binary search.
\end{proof}

\section{Discussion}
We examined the \maxmincs problem of selecting $k$ mutually disjoint convex subsets of a given set of points $P$ such that the cardinality of the smallest set is maximized.  
We described an algorithm that solves \maxmincs for small values of $k$, showed \maxmincs is NP-hard for an arbitrary $k$, and proved \maxmincs to be fixed parameter tractable when parameterized by the number of points interior to the convex hull.  
The problem is also solvable in polynomial time for specific large values of $k$.
If $k > n/4$, then some subset has cardinality at most three; a solution is found trivially by arbitrarily partitioning $P$ into $k$ subsets of size $\lfloor n/k\rfloor$ or $\lceil n/k\rceil$.
If $k \in \{\lfloor n/5\rfloor +1, \ldots, n/4\}$ then some subset has cardinality at most four. 
As discussed in Section~\ref{sec:related},
K\'{a}rolyi~\cite{DBLP:journals/dam/Karolyi01} characterized necessary and sufficient conditions for a set of $n$ points in general position to admit a partition into $k=n/4$ convex quadrilaterals, and gave an $O(n \log n)$-time algorithm to decide whether such a partition exists; if no such partition exists, then some set must contain at most three points, which can be solved as described above. It remains open to determine whether \maxmincs can be solved in polynomial time for all $k \in \Theta(n)$.

As a direction for future research, a natural question is to establish a good lower bound on the time required to solve these problems for small fixed values of $k$. In particular, is the $O(n^3)$-time algorithm of Chv\'{a}tal and Klincsek \cite{chvatal1980} optimal for the case $k=1$? Note that our algorithm has time $O(n^8)$ when $k=1$. 
It would also be interesting to examine whether a fixed-parameter tractable algorithm exists for \maxmincs when parameterized by $k$, and to find approximation algorithms for \maxmincs when $k$ is an arbitrary input parameter, with running time polynomial in $n$ and $k$.

\bibliographystyle{splncs04}
\bibliography{ref2}

\begin{thebibliography}{10}
\providecommand{\url}[1]{\texttt{#1}}
\providecommand{\urlprefix}{URL }
\providecommand{\doi}[1]{https://doi.org/#1}

\bibitem{DBLP:journals/corr/cs-CG-0210003}
Arkin, E.M., Fekete, S.P., Hurtado, F., Mitchell, J.S.B., Noy, M.,
  Sacrist{\'{a}}n, V., Sethia, S.: On the reflexivity of point sets. Discrete
  and Computational Geometry  \textbf{25},  139--156 (2003)

\bibitem{DBLP:conf/compgeom/AvisR85}
Avis, D., Rappaport, D.: Computing the largest empty convex subset of a set of
  points. In: Proc.\ Symposium on Computational Geometry (SoCG). pp. 161--167
  (1985)

\bibitem{DBLP:conf/cccg/BrodenHN01}
Brod{\'{e}}n, B., Hammar, M., Nilsson, B.J.: Guarding lines and 2-link polygons
  is apx-hard. In: Proc.\ Canadian Conference on Computational Geometry (CCCG).
  pp. 45--48 (2001)

\bibitem{chvatal1980}
Chv\'{a}tal, V., Klincsek, G.: Finding largest convex subsets. Congressus
  Numerantium  \textbf{29},  453--460 (1980)

\bibitem{DBLP:conf/jcdcg/DingHUX02}
Ding, R., Hosono, K., Urabe, M., Xu, C.: Partitioning a planar point set into
  empty convex polygons. In: Proc.\ Discrete and Computational Geometry,
  Japanese Conference (JCDCG). LNCS, vol.~2866, pp. 129--134 (2002)

\bibitem{DBLP:journals/jcss/EdelsbrunnerG89}
Edelsbrunner, H., Guibas, L.J.: Topologically sweeping an arrangement. J.
  Comput. Syst. Sci.  \textbf{38}(1),  165--194 (1989)

\bibitem{Eppstein}
Eppstein, D.: Forbidden Configurations in Discrete Geometry. Cambridge
  University Press (2018)

\bibitem{DBLP:journals/dcg/EppsteinE94}
Eppstein, D., Erickson, J.: Iterated nearest neighbors and finding minimal
  polytopes. Disc. Comp. Geom.  \textbf{11},  321--350 (1994)

\bibitem{DBLP:journals/dcg/EppsteinORW92}
Eppstein, D., Overmars, M.H., Rote, G., Woeginger, G.J.: Finding minimum area
  k-gons. Disc. Comp. Geom.  \textbf{7},  45--58 (1992)

\bibitem{erdos}
{Erd\H{o}s}, P., Szekeres, G.: A combinatorial problem in geometry. Compositio
  Mathematica  \textbf{2},  463--470 (1935)

\bibitem{DBLP:journals/ejc/ErdosTV96}
{Erd\H{o}s}, P., Tuza, Z., Valtr, P.: Ramsey-remainder. Eur. J. Comb.
  \textbf{17}(6),  519--532 (1996)

\bibitem{DBLP:journals/comgeo/GajentaanO12}
Gajentaan, A., Overmars, M.H.: On a class of {$O(n^2)$} problems in
  computational geometry. Comp. Geom.  \textbf{45}(4),  140--152 (2012)

\bibitem{DBLP:conf/esa/GiannopoulosKW13}
Giannopoulos, P., Knauer, C., Werner, D.: On the computational complexity of
  {Erd{\H{o}}s-Szekeres} and related problems in {$R^3$}. In: Proc.\ European
  Symposium on Algorithms (ESA). pp. 541--552 (2013)

\bibitem{DBLP:conf/fct/Gonzalez-Aguilar19}
Gonz{\'{a}}lez{-}Aguilar, H., Orden, D., P{\'{e}}rez{-}Lantero, P., Rappaport,
  D., Seara, C., Tejel, J., Urrutia, J.: Maximum rectilinear convex subsets.
  In: Proc.\ Symposium on Fundamentals of Computation Theory (FCT). pp.
  274--291 (2019)

\bibitem{DBLP:journals/orl/HulettWW08}
Hulett, H., Will, T.G., Woeginger, G.J.: Multigraph realizations of degree
  sequences: Maximization is easy, minimization is hard. Oper. Res. Lett.
  \textbf{36}(5),  594--596 (2008)

\bibitem{DBLP:journals/dam/Karolyi01}
K{\'{a}}rolyi, G.: Ramsey-remainder for convex sets and the
  {Erd\H{o}s}-szekeres theorem. Discret. Appl. Math.  \textbf{109}(1-2),
  163--175 (2001)

\bibitem{kirkpatrick1986}
Kirkpatrick, D.G., Seidel, R.: The ultimate planar convex hull algorithm? SIAM
  J. Comput.  \textbf{15}(1),  287--299 (1986)

\bibitem{DBLP:conf/icalp/KumarAR00}
Kumar, V.S.A., Arya, S., Ramesh, H.: Hardness of set cover with intersection 1.
  In: Proc.\ International Colloquium on Automata, Languages and Programming
  (ICALP). pp. 624--635 (2000)

\bibitem{DBLP:journals/orl/MegiddoT82}
Megiddo, N., Tamir, A.: On the complexity of locating linear facilities in the
  plane. Oper. Res. Lett.  \textbf{1}(5),  194--197 (1982)

\bibitem{Morris}
Morris, W., Soltan, V.: The {Erd\H{o}s-Szekeres} problem on points in convex
  position - a survey. Bulletin of The American Mathematical Society
  \textbf{37},  437--459 (10 2000)

\bibitem{DBLP:journals/dam/Urabe96}
Urabe, M.: On a partition into convex polygons. Discret. Appl. Math.
  \textbf{64}(2),  179--191 (1996)

\end{thebibliography}
\end{document}